\documentclass[a4paper]{amsart}

\usepackage{amsmath}
\usepackage{amssymb}
\usepackage{amsthm}
\usepackage{graphicx}
\usepackage{bbm}
\usepackage{hyperref}

\newtheorem{theorem}{Theorem}
\newtheorem{proposition}[theorem]{Proposition}

\newcommand{\expec}{\mathbf{E}}

\newcommand{\prob}{\mathbf{P}}

\title[Digital double barrier options]{Digital double barrier options: Several barrier periods and structure floors}

\author{S\"uhan Altay}
\author{Stefan Gerhold}
\author{Karin Hirhager}

\address{Vienna University of Technology, Wiedner Hauptstra\ss{}e 8--10,
A-1040 Vienna, Austria}

\email{sgerhold at fam.tuwien.ac.at}

\thanks{This work was financially supported by the Christian Doppler Research
Association (\href{http://www.cdg.ac.at/}{CDG}).
The authors gratefully acknowledge the fruitful collaboration
and support by the \href{http://www.bankaustria.at/}{Bank Austria},
\href{http://www.cor.fja.com/}{COR \& FJA}
and the Austrian Federal Financing Agency
(\href{http://www.oebfa.co.at/}{\"OBFA}) through CDG and the
CD-Laboratory for Portfolio Risk Management (PRisMa Lab)
\url{http://www.prismalab.at/}.
\\ We thank Richard~C.\ Bradley and Antoine Jacquier for valuable comments and discussions.
}

\date{\today}

\begin{document}

\begin{abstract}
We determine the price of digital double barrier options with an arbitrary number
of barrier periods in the Black-Scholes model. This means that the barriers are active
during some time intervals, but are switched off in between.
As an application, we calculate
the value of a structure floor for structured notes whose individual
coupons are digital double barrier options. This value can also be approximated
by the price of a corridor put.
\end{abstract}

\keywords{Double barrier option, structure floor, occupation time, corridor option}

\subjclass[2010]{Primary: 91G20; Secondary: 60J65} 

\maketitle

\section{Introduction}

We consider digital double barrier options with an arbitrary number
of barrier periods.
This means that the holder receives the payoff only if the underlying
stays between the two barriers in certain specified time intervals.
While such contracts might make sense by themselves
(as a weather or energy derivative with seasonal barriers, say), our motivation is
to use them for the pricing of certain structured notes with several coupons.
Such trades often feature
an aggregate floor at the final coupon date, which increases the total payoff to a guaranteed 
amount if the sum of the coupons is less than this amount.
Pricing this terminal premium requires the law of the sum of the coupons, which
can be recovered from its moments.
If the individual coupons of the note are digital barrier options,
then these moments can be computed from the prices of options
of the kind described above, where the sets of barrier
periods are subsets of the coupon periods of the note.

Recall that Monte Carlo pricing of barrier contracts is tricky,
because the discretization produces a downward bias for the barrier hitting
probability. For single barrier options, this difficulty can be overcome
using the explicit law of the maximum of the Brownian bridge~\cite{AnBr96,BeDyZh97}.
For double barrier options, the exit probability of the Brownian
bridge is not known; see Baldi et al.~\cite{BaCaIo99} for an approximate
approach using sample path large deviations. These numerical
challenges led us to investigate exact valuation formulas.

The paper is structured as follows. In Section~\ref{se:one per} we define
the payoffs we are interested in and price them for a single
barrier period. Section~\ref{se:arb per} extends the result to arbitrarily
many periods of active barriers. Our main application, namely
the pricing of structure floors, is presented in Section~\ref{se:struct}.
Since our exact pricing formula is fairly involved, we consider an
asymptotic approximation for a large number of periods in Section~\ref{se:corr}.

\section{Preliminaries and pricing for one period}\label{se:one per}

We assume that the underlying $(S_t)_{t\geq0}$ has the risk-neutral dynamics
\[
  dS_t/S_t = r dt + \sigma dW_t
\]
with constant interest rate $r>0$, volatility $\sigma>0$
and a standard Brownian motion~$W$.
Consider a digital barrier option with two
barriers~$B_{\mathrm{low}}$
and~$B_{\mathrm{up}}$ that are
activated at time $T_0>0$ and stay active for a time period of length~$P>0$.
At maturity $T_0+P$, the payoff is one unit of currency if the underlying
has stayed between the two barriers:
\begin{equation}\label{eq:C_1}
  C_1 := \mathbf{1}_{\{B_{\mathrm{low}} < S_t < B_{\mathrm{up}},\ 
    t \in [T_{0},T_{0}+P] \}}.
\end{equation}
Let us denote the price of this ``one-period
double barrier digital'' by
\begin{equation}\label{eq:def BD1}
  BD(S_t,t;\{T_0\},P,B_{\mathrm{low}},B_{\mathrm{up}},r):=
  e^{-r(T_{0}+P)} \expec[ C_1],
\end{equation}
where $\expec$ is the expectation w.r.t.\ the pricing measure~$\mathbf{P}$.
In the terminology of Hui~\cite{Hu97},
this is a \emph{rear-end} barrier option, because the two barriers are alive only
towards the end of the contract, namely between~$T_0$ and maturity $T_0+P$.
Hui~\cite{Hu97} has determined the price for a barrier \emph{call} of this kind.
The digital case is a simple modification, but we go through it to prepare
the calculation of the price for several barrier periods (see Section~\ref{se:arb per}).
The value function
\[
  f(S,t) := BD(S,t;\{T_0\},P,B_{\mathrm{low}},B_{\mathrm{up}},r)
\]
satisfies the Black-Scholes PDE
\[
\frac{\partial f}{\partial t} + \frac{1}{2} \sigma^{2} S^{2} \frac{\partial^{2} f}{\partial S^{2}} + r S \frac{\partial f}{\partial S} - rf = 0
\]
with the terminal condition $f(S,T_0+P)=1$, for $S \in (B_{\mathrm{low}},
B_{\mathrm{up}})$, and the boundary conditions
$f(B_{\mathrm{low}},t) = f(B_{\mathrm{up}},t) =0$ for $t \in [T_0, T_0+P]$.
We use the standard transformation $f(S,t) = e^{\alpha x + \beta \tau} U(x,\tau)$,
where
\begin{align} 
  x &:= \log(S/B_{\mathrm{low}}), \qquad \tau := \tfrac12 \sigma^2(T_0+P - t),
    \label{eq:coord} \\
  \alpha &:= - \frac{1}{2} \left(\frac{2}{\sigma^{2}}r - 1\right), \qquad
  \beta := -\frac{2r}{\sigma^{2}} - \alpha^{2}, \notag
\end{align}
to transform the Black-Scholes PDE into the heat equation
\begin{equation}\label{eq:heat}
  \frac{\partial^{2} U}{\partial {x}^{2}} = \frac{\partial U}{\partial \tau}.
\end{equation}
The time points $(0,T_0,T_0+P)$ are thus converted to $(\tfrac12 \sigma^2(T_0+P),
p,0)$, where $p:=\tfrac12 \sigma^2 P$ is the barrier period length in the new time scale.
The boundary conditions in the new coordinates are
\begin{equation}\label{eq:boundary}
  U(0,\tau)=U(L,\tau)=0, \qquad \tau \in [0, p],
\end{equation}
where $L:=\log(B_{\mathrm{up}}/B_{\mathrm{low}})$.
The terminal condition translates to the initial condition
\begin{equation}\label{eq:init}
  U(x,0) = e^{-\alpha x}, \qquad x \in (0,L).
\end{equation}
\begin{proposition}\label{prop:one per}
For $0<t<T_0$, the price of a barrier digital with barrier period $[T_0,T_0+P]$ and
payoff~$C_1$ at~$T_0+P$ (see~\eqref{eq:C_1}) is
\begin{multline}\label{eq:one per}
  BD(S,t;\{T_0\},P,B_{\mathrm{low}},B_{\mathrm{up}},r)
    = \sqrt{2\pi} \left( \frac{S}{B_{\mathrm{low}}} \right)^\alpha
    \sum_{k=1}^{\infty} k \frac{1-(-1)^{k} e^{-\alpha L} }{\alpha^{2} L^{2} + k^{2} \pi^{2}} 
    e^{-(\frac{k\pi}{L})^{2}p+\beta \tau}  \\
  \cdot
    \int_{-\frac{x}{\sqrt{2(\tau-p)}}}^{\frac{L-x}{\sqrt{2(\tau-p)}}}   
    \sin\left(\frac{k\pi}{L}(x+y \sqrt{2(\tau-p)})\right) e^{-y^2/2}dy.
\end{multline}
\end{proposition}
\begin{proof}
  We have to solve the problem \eqref{eq:heat}--\eqref{eq:init}.
  First consider the rectangle $(0,L)\times(0,p)$.
  There the solution can be found by separation of variables~\cite[Section~4.1]{Ev98}:
  \begin{equation}\label{U rect}
    U(x, \tau) = \sum_{k=1}^{\infty} b_{k} \sin\left(\frac{k\pi}{L}x\right)
      e^{-(\frac{k\pi}{L})^{2}\tau}, \qquad (x,\tau) \in
      (0,L)\times(0,p),
  \end{equation}
  where
  \[
    b_{k} := \frac{2}{L} \int_{0}^{L} e^{-\alpha x_1} \sin\left(\frac{k\pi}{L}x_1\right) dx_1
      = 2k\pi \frac{1-(-1)^k e^{-\alpha L}}{\alpha^2 L^2+k^2\pi^2}
  \]
  are the Fourier coefficients of the boundary function $U(x,0)=e^{-\alpha x}$.
  At $\tau=p$, the solution is given by~\eqref{U rect}
  for $0<x<L$ and vanishes otherwise. Inserting $\tau=p$
  into~\eqref{U rect} yields
  \begin{equation}\label{eq:U middle}
    U(x,p) =
    \begin{cases}
    \sum_{k=1}^{\infty} 2k\pi \frac{1-(-1)^{k} e^{-\alpha L} }{\alpha^{2} L^{2} + k^{2}
      \pi^{2}} \sin(\frac{k\pi}{L}x)
      e^{-(\frac{k\pi}{L})^{2} p},& 0<x<L \\
      0, & x\leq 0 \ \text{or}\ x\geq L.
    \end{cases}
  \end{equation}
  Now we solve for~$U$ in the region $\mathbb{R}\times(p,
  \tfrac12 \sigma^2(T_0+P))$. There are no boundary conditions here, 
  since the barriers are not active in the interval $(0,T_0)$
  (in the original time scale).
  The solution is found by convolving the initial condition~\eqref{eq:U middle}
  with the heat kernel~\cite[2.3.1.b]{Ev98}:
  \begin{align}
    U(x,\tau) &= \frac{1}{\sqrt{2\pi}} \int_{-\infty}^{\infty} U(x+y  
    \sqrt{2(\tau-p)}, p)e^{-y^2/2}dy \notag \\
    &= \frac{1}{\sqrt{2\pi}} 
    \int_{-\frac{x}{\sqrt{2(\tau-p)}}}^{\frac{L-x}{\sqrt{2(\tau-p)}}} 
    U(x+y \sqrt{2(\tau-p)}, p)e^{-y^2/2}dy. \label{eq:conv}
  \end{align}
  Inserting~\eqref{eq:U middle} and rearranging yields~\eqref{eq:one per}.
\end{proof}

\section{Double barrier digitals with
arbitrarily many periods}\label{se:arb per}

For~$n$ tenor dates
\[
  0 < T_0 < \dots < T_{n-1}
\]
and a fixed period length $P>0$, we consider a contract that
pays one unit of currency at time~$T_{n-1}+P$, if the underlying
has remained between the two barriers~$B_{\mathrm{low}}$
and~$B_{\mathrm{up}}$ during each of the time intervals $[T_i,T_i+P]$,
$i=0,\dots,n-1$. By the risk-neutral pricing formula, the price of this ``multi-period
double barrier digital'' is given by
\begin{equation}\label{eq:def BD}
  BD(S_t,t;\{T_0,\dots,T_{n-1}\},P,B_{\mathrm{low}},B_{\mathrm{up}},r):=
  e^{-r(T_{n-1}+P)} \expec\left[\prod_{i=1}^n C_i\right],
\end{equation}
where
\begin{equation*}
  C_i := \mathbf{1}_{\{B_{\mathrm{low}} < S_t < B_{\mathrm{up}},\ 
    t \in [T_{i-1},T_{i-1}+P] \}}.
\end{equation*}
To calculate the price,
we use the coordinate change~\eqref{eq:coord} again. The~$n$ barrier periods
$[T_i,T_i+P]$ are mapped to $[\tau_i,\tau_i+p]$, where
\[
  \tau_i := \tfrac12 \sigma^2(T_{n-1}-T_{i-1}), \qquad i = n,\dots,1,
\]
are the images of the barrier period endpoints under
the coordinate change (see Figure~\ref{fig:barriers}).
\begin{figure}[h]
  \includegraphics[scale=1.0]{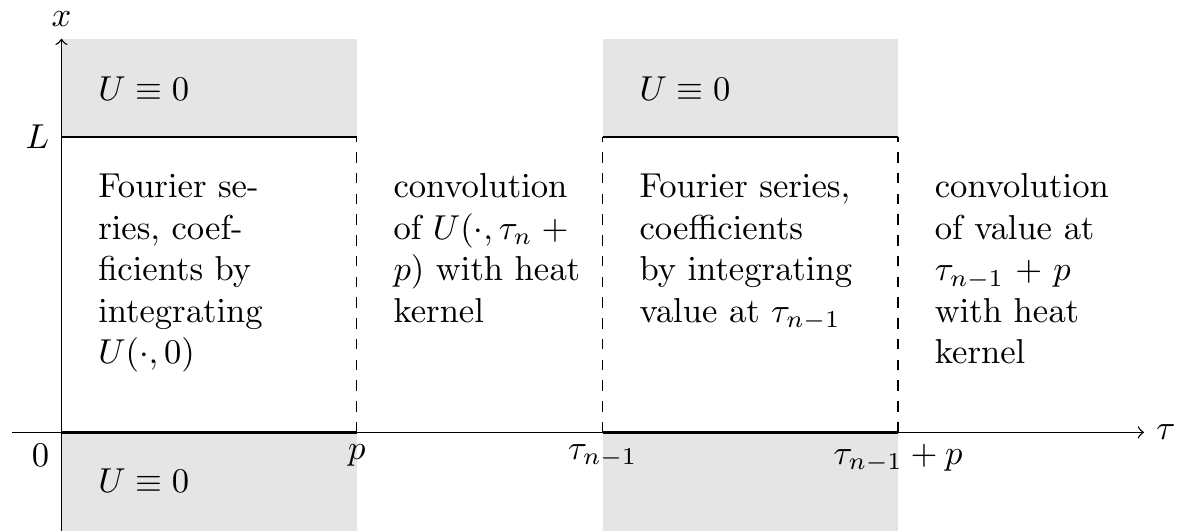}
  \caption{Solving the boundary value problem for an arbitrary number of barrier periods.}
  \label{fig:barriers}
\end{figure}
Define the following auxiliary functions:
\begin{multline*}
  h_j(k_1,\dots,k_{j+1}; x_1,\dots,x_{j+1}; y_1,\dots,y_{j+1}; x,\tau) \\
  := \tfrac{1}{\sqrt{2\pi}} e^{-y_{j+1}^2/2}\ \mathbf{1}_{\left[-\tfrac{x}{\sqrt{2(\tau-(\tau_{n-j}+p))}},
  \tfrac{L-x}{\sqrt{2(\tau-(\tau_{n-j}+p))}}\right]}(y_{j+1}) \\
  \cdot g_j(k_1,\dots,k_{j+1}; x_1,\dots,x_{j+1};
    y_1,\dots,y_j; x + y_{j+1}\sqrt{2(\tau-(\tau_{n-j}+p))} , \tau_{n-j}+p)
\end{multline*}
and
\begin{multline*}
  g_j(k_1,\dots,k_{j+1}; x_1,\dots,x_{j+1}; y_1,\dots,y_j; x,\tau) \\
  := \tfrac{2}{L}  \sin \tfrac{k_{j+1} \pi x_{j+1}}{L} \sin \tfrac{k_{j+1} \pi x}{L}
    e^{-(k_{j+1}\pi/L)^2(\tau-\tau_{n-j})} \\
  \cdot h_{j-1}(k_1,\dots,k_j; x_1,\dots,x_j; y_1,\dots,y_j; x_{j+1},\tau_{n-j}),
\end{multline*}
with the recursion starting at
\begin{equation}\label{eq:g0}
  g_0(k_1;x_1;;x,\tau) := \tfrac{2}{L} e^{-\alpha x_1}
  \sin \tfrac{k_1 \pi x_1}{L} \sin \tfrac{k_1 \pi x}{L}
    e^{-(k_1\pi/L)^2\tau}.
\end{equation}
The following theorem contains our pricing formula.
The first formula~\eqref{eq:U1}
is for time points inside a barrier period, whereas
the second expression~\eqref{eq:U2} holds for valuation times where the barriers
are not active.
\begin{theorem}\label{thm:main}
  The value function~\eqref{eq:def BD} equals $e^{\alpha x + \beta \tau} U(x,\tau)$,
  where for $0\leq j < n$, $\tau_{n-j} \leq \tau \leq \tau_{n-j} + p$, $0<x<L$, we have
  \begin{multline}\label{eq:U1}
  U(x,\tau) = \underbrace{\int_{-\infty}^\infty \dots \int_{-\infty}^\infty}_{j}
    \underbrace{\int_0^L \dots \int_0^L}_{j+1}
    \sum_{k_1=0}^\infty \dots \sum_{k_{j+1}=0}^\infty  \\
     g_j(k_1,\dots,k_{j+1}; x_1,\dots,x_{j+1}; y_1,\dots,y_j; x,\tau)
    dx_1 \dots dx_{j+1} dy_1 \dots dy_{j},  
  \end{multline}
  whereas for $0\leq j <n$, $\tau_{n-j} + p < \tau < \tau_{n-(j+1)}$
  (with $\tau_0:=\infty$), $x\in\mathbb{R}$, we have
  \begin{multline}\label{eq:U2}
    U(x,\tau) = \underbrace{\int_{-\infty}^\infty \dots \int_{-\infty}^\infty}_{j+1}
    \underbrace{\int_0^L \dots \int_0^L}_{j+1}
    \sum_{k_1=0}^\infty \dots \sum_{k_{j+1}=0}^\infty  \\
     h_j(k_1,\dots,k_{j+1}; x_1,\dots,x_{j+1}; y_1,\dots,y_{j+1}; x,\tau)
    dx_1 \dots dx_{j+1} dy_1 \dots dy_{j+1}. 
  \end{multline}
\end{theorem}
\begin{proof}
  The idea is to iterate the argument of Proposition~\ref{prop:one per}
  (see Figure~\ref{fig:barriers}). We use separation of
  variables in the barrier periods, and convolution with the heat
  kernel for the periods in between. The required initial condition at
  the left boundary comes from the previous step of the iteration
  (for $j=0$ also from the payoff, of course).
  
  For $j=0$, formula~\eqref{eq:U1} is identical to~\eqref{U rect}.
  To show~\eqref{eq:U2} for $j=0$, let $p<\tau<\tau_{n-1}$
  (recall that $\tau_n=0$) and $x\in\mathbb{R}$,
  and use~\eqref{eq:conv} and~\eqref{U rect} to obtain
  \begin{align*}
    U(x,\tau) &= \frac{1}{\sqrt{2\pi}} 
      \int_{-\frac{x}{\sqrt{2(\tau-p)}}}^{\frac{L-x}{\sqrt{2(\tau-p)}}} 
      U(x+y_1 \sqrt{2(\tau-p)}, p)e^{-y_1^2/2}dy_1 \\
    &= \frac{1}{\sqrt{2\pi}} \int_{-\infty}^\infty
      \mathbf{1}_{\left[-\tfrac{x}{\sqrt{2(\tau-p)}},
      \tfrac{L-x}{\sqrt{2(\tau-p)}}\right]}(y_{1}) \\
      & \qquad \qquad
      \int_0^L \sum_{k_1=0}^\infty g_0(k_1;x_1;;x+y_1\sqrt{2(\tau-p)},p)
      e^{-y_1^2/2}dx_1 dy_1 \\
    &= \int_{-\infty}^\infty \int_0^L \sum_{k_1=0}^\infty
      h_0(k_1;x_1;y_1;x,\tau) dx_1 dy_1.
  \end{align*}
  This is~\eqref{eq:U2} for $j=0$.
  
  Next consider a rectangle
  \begin{equation}\label{eq:rect j}
    (\tau,x) \in (\tau_{n-j}, \tau_{n-j}+p) \times (0,L), \qquad
      1\leq j < n.
  \end{equation}
  At the left boundary, the solution is $x_{j+1}\mapsto U(x_{j+1},\tau_{n-j})$.
  By the induction hypothesis, it equals~\eqref{eq:U2} with~$j$
  replaced by $j-1$:
    \begin{multline}\label{eq:U2 j-1}
    U(x_{j+1},\tau_{n-j}) = \underbrace{\int_{-\infty}^\infty \dots 
    \int_{-\infty}^\infty}_{j}
    \underbrace{\int_0^L \dots \int_0^L}_{j}
    \sum_{k_1=0}^\infty \dots \sum_{k_{j}=0}^\infty  \\
     h_{j-1}(k_1,\dots,k_{j}; x_1,\dots,x_{j}; y_1,\dots,y_{j}; x_{j+1},\tau_{n-j})
    dx_1 \dots dx_{j} dy_1 \dots dy_{j}. 
  \end{multline}
  The solution in the rectangle~\eqref{eq:rect j} is thus obtained by
  separation of variables as
  \begin{equation}\label{eq:sol j rect}
    U(x,\tau) = \sum_{k_{j+1}=0}^\infty b_{k_{j+1}} 
      \sin\left(\frac{k_{j+1}\pi}{L}x\right)
      e^{-(\frac{k_{j+1}\pi}{L})^{2}(\tau-\tau_{n-j})},
  \end{equation}
  where
  \begin{equation}\label{eq:fourier}
     b_{k_{j+1}} := \frac{2}{L} \int_{0}^{L} U(x_{j+1},\tau_{n-j})
       \sin\left(\frac{k_{j+1}\pi}{L}x_{j+1}\right) dx_{j+1}
  \end{equation}
  denote now the Fourier coefficients of $x_{j+1}\mapsto U(x_{j+1},\tau_{n-j})$.
  Inserting~\eqref{eq:U2 j-1} into~\eqref{eq:fourier}
  and then~\eqref{eq:fourier} into~\eqref{eq:sol j rect} yields~\eqref{eq:U1},
  by the definition of~$g_j$.
  
  Finally, consider a strip
  \begin{equation}\label{eq:strip}
    (\tau,x) \in (\tau_{n-j}+p, \tau_{n-(j+1)}) \times \mathbb{R}, \qquad
      1\leq j < n.
  \end{equation}
  At the left boundary, we use~\eqref{eq:U1} as induction hypothesis.
  The solution thus vanishes for $x\notin(0,L)$, and for $\tau=\tau_{n-j}+p$
  and $x\in(0,L)$ it is
  \begin{multline}\label{eq:U1 again}
    U(x,\tau_{n-j}+p) = \underbrace{\int_{-\infty}^\infty \dots \int_{-\infty}^\infty}_{j}
    \underbrace{\int_0^L \dots \int_0^L}_{j+1}
    \sum_{k_1=0}^\infty \dots \sum_{k_{j+1}=0}^\infty \\
     g_j(k_1,\dots,k_{j+1}; x_1,\dots,x_{j+1}; y_1,\dots,y_j; x,\tau_{n-j}+p)
    dx_1 \dots dx_{j+1} dy_1 \dots dy_{j}.  
  \end{multline}
  As above, the solution in the strip~\eqref{eq:strip} is found by convolution
  with the heat kernel:
  \begin{multline*}
    U(x,\tau) 
    = \frac{1}{\sqrt{2\pi}} \int_{-\infty}^\infty
      \mathbf{1}_{\left[-\tfrac{x}{\sqrt{2(\tau-(\tau_{n-j}-p))}},
      \tfrac{L-x}{\sqrt{2(\tau-(\tau_{n-j}-p))}}\right]}(y_{j+1}) \\
      \qquad U(x+y_{j+1}
      \sqrt{2(\tau-(\tau_{n-j}-p))}, \tau_{n-j}-p)e^{-y_{j+1}^2/2}dy_{j+1}.
  \end{multline*}
  Now insert~\eqref{eq:U1 again}, with~$x$ replaced by $x+y_{j+1}
  \sqrt{2(\tau-(\tau_{n-j}+p))}$, and use the definition of~$h_j$
  to conclude~\eqref{eq:U2}.
\end{proof}
Note that Proposition~\ref{prop:one per} corresponds to~\eqref{eq:U2} for $j=0$.
As seen there, the integral $\int_0^L dx_1$ can be done in closed form.
We have not included  this evaluation
in Theorem~\ref{thm:main} to increase its readability.

If a different option (a call, say) with the same barrier conditions
is to be priced instead of a digital payoff,
the quantity $e^{-\alpha x_1}$ in~\eqref{eq:g0} should be replaced
by the appropriate payoff $U(x_1,0)$.

\section{Structure floors}\label{se:struct}

In this section we assume that our tenor structure satisfies
$T_{i-1}+P=T_i$ for $1\leq i< n$, and define $T_n:=T_{n-1}+P$.
We consider a structured note with~$n$
coupons, where the $i$-th coupon consists of a payment of
\begin{equation}\label{eq:C_i}
  C_i = \mathbf{1}_{\{B_{\mathrm{low}} < S_t < B_{\mathrm{up}},\  t \in [T_{i-1},T_i] \}},
  \qquad 1\leq i\leq n,
\end{equation}
at time $T_i$.
These coupons can be priced by Proposition~\ref{prop:one per} (replace
$T_0$ by $T_{i-1}$).
In addition, the holder receives the terminal premium
\begin{equation}\label{eq:struc floor}
  \left(F - \sum_{i=1}^n C_i\right)^+
\end{equation}
at $T_n$, where $F>0$. This means that the aggregate payoff $A := \sum_{i=1}^n C_i$
of the note is floored at~$F$, which is a popular feature of structured notes.
While the individual coupons are straightforward to valuate, it is less
obvious how to get a handle on the law of~$A$. We now show that this law
is linked to barrier options with several barrier periods.
Indeed, the following result is based on the fact that the moments
\begin{equation}\label{eq:momA}
  \expec[A^\nu] = \sum_{i=0}^n i^\nu \prob[A=i], \qquad 1\leq \nu < n,
\end{equation}
of~$A$ are linear combinations
of multi-period double barrier option prices, with coefficients
\begin{equation}\label{eq:cJ}
  c(\nu,J) := \sum_{\substack{0\leq i_1,\dots,i_n\leq \nu \\ \mathrm{supp}
    (\mathbf{i}) = J}} \binom{\nu}{i_1,\dots,i_n}, \qquad J\subseteq \{1,\dots,n\}.
\end{equation}
(The notation $\mathrm{supp}(\mathbf{i}) = J$ means that~$J$ is
the set of indices such that the corresponding components of the vector
$\mathbf{i}=(i_1,\dots,i_n)$ are non-zero.)
\begin{theorem}\label{thm:floor}
  The price of the structure floor~\eqref{eq:struc floor} at time~$t<T_0$ can be expressed as
  \begin{equation}\label{eq:tp price}
    e^{-r(T_n-t)} \expec[(F-A)^+]
      = e^{-r(T_n-t)} \sum_{i=0}^{n \wedge \lfloor F \rfloor} (F-i) \prob[A=i],
  \end{equation}
  where
  \begin{equation}\label{eq:A=n}
    \prob[A=n] = BD(S_t,t;\{T_0\},T_n-T_0,B_{\mathrm{low}},B_{\mathrm{up}},0).
  \end{equation}
  The other point masses $\prob[A=i]$ in~\eqref{eq:tp price}
  can be recovered from the moments of~$A$ by solving~\eqref{eq:momA}
  (including $\nu=0$, of course).
  The moments in turn can be computed from barrier digital prices by
  ($1\leq \nu<n$)
  \begin{equation}\label{eq:A moments}
    \expec[A^\nu] = \sum_{J\subseteq \{1,\dots,n\}}
    c(\nu,J) \cdot
    BD(S_t,t;\{T_j:j\in J\},P,B_{\mathrm{low}},B_{\mathrm{up}},0),
  \end{equation}
  where the coefficients $c(\nu,J)$ are defined in~\eqref{eq:cJ}.
\end{theorem}
\begin{proof}
  The expression~\eqref{eq:tp price} is clear. The event in~\eqref{eq:A=n}
  means that all of the~$n$ coupons~\eqref{eq:C_i} are paid.
  By our assumption that $T_i = T_{i-1}+P$,
  its risk-neutral probability is the (undiscounted) price of a double barrier digital
  with one barrier period $[T_0,T_n]$, which yields~\eqref{eq:A=n}.
  To prove~\eqref{eq:A moments},
  we calculate
  \begin{align*}
  \expec[A^\nu] = \expec\left[\left(\sum_{i=1}^n C_i\right)^\nu\right] &= \sum_{i_1,\dots,i_n} \binom{\nu}{i_1,\dots,i_n} \expec[C_1^{i_1} \dots C_n^{i_n}] \\
  &= \sum_{i_1,\dots,i_n} \binom{\nu}{i_1,\dots,i_n} \expec\bigg[\prod_{\substack{j=1 \\ i_j>0}}^n C_j\bigg] \\
  &= \sum_{J\subseteq \{1,\dots,n\}} \Bigg( \sum_{\substack{i_1,\dots,i_n \\ \mathrm{supp}
    (\mathbf{i}) = J}} \binom{\nu}{i_1,\dots,i_n} \Bigg) \expec\Big[\prod_{j\in J} C_j\Big].
  \end{align*}
  Now observe that $\prod_{j\in J} C_j$ is the payoff of a double barrier digital
  with barrier periods $[T_j,T_j+P]$ for $j\in J$.
\end{proof}
When calculating the value $BD$ in~\eqref{eq:A moments} for,
say, $J=\{1,2,4,5,6\}$, the adjacent barrier periods should be concatenated:
Do not compute the price for five barrier periods of length~$P$, but rather
for two periods with lengths $2P$ and $3P$. We did not include this obvious
extension (barrier periods of variable length) in Theorem~\ref{thm:main} in order
not to complicate an already heavy notation.

\section{Approximation by a corridor option}\label{se:corr}

Theorems~\ref{thm:main} and~\ref{thm:floor} express the price of the structure floor~\eqref{eq:struc floor} by iterated sums and integrals. Due to the factors
of order $e^{-k_j^2}$, the infinite series $\sum_{k_j}$
may be truncated after just a few terms.
Still, numerical quadrature may be too involved for a large number of coupons, so
we present an approximation. Let us fix a maturity~$T=T_n$ and assume
that the~$n$ coupon periods
\[
  \mathcal{T}_i^n := [\tfrac{i-1}{n}T, \tfrac{i}{n}T], \qquad 1\leq i\leq n,
\]
have length $T/n$. For large~$n$, the proportion of intervals during
which the underlying stays inside the barrier interval
\[
  \mathcal{B}:=[B_{\mathrm{low}}, B_{\mathrm{up}}]
\]
is similar to the proportion of time that the underlying spends inside~$\mathcal{B}$,
i.e., the occupation time.
This is made precise in the following result, which
holds not only for the Black-Scholes model,
but for virtually any continuous model.
Note that the level sets of geometric Brownian motion have a.s.\ measure
zero (cf.~\cite[Theorem~2.9.6]{KaSh91}).
\begin{theorem}\label{thm:corr}
  Let $(S_t)_{t\geq0}$ be a continuous stochastic process such that for each
  real~$c$ the level set $\{t: S_t=c\}$ has a.s.\ Lebesgue measure zero.
  Then we have a.s.
  \[
    \lim_{n\to\infty} \frac1n \sum_{i=1}^n \mathbf{1}_{\{S_t \in
      \mathcal{B}\ \forall t\in \mathcal{T}_i^n\}}
      = \int_0^T \mathbf{1}_{\mathcal{B}}(S_t)dt.
  \]
\end{theorem}
\begin{proof}
  For $1\leq i\leq n$,
  define processes $(X_{ni}(t))_{0\leq t\leq T}$ by
  \[
    X_{ni}(t) :=
    \begin{cases}
      1 & \text{if}\ t\in \mathcal{T}_i^n\ \text{and}\ S_u\in \mathcal{B}\
        \forall u\in \mathcal{T}_i^n \\
      0 & \text{otherwise}.
    \end{cases}
  \]
  Put $X_n:=\sum_{i=1}^n X_{ni}$.
  We claim that, a.s., the function $X_n(\cdot)$ converges pointwise
  on the set $[0,T]\setminus \{t: S_t = B_{\mathrm{low}}\ \text{or}\
  S_t=B_{\mathrm{up}}\}$, with limit $\mathbf{1}_{\mathcal{B}}(S_\cdot)$.
  Indeed, if $t\in[0,T]$ is such that $S_t\notin \mathcal{B}$,
  then $X_{n}(t)=0$ for all~$n$. If, on the other hand,
  $S_t\in \mathrm{int}(\mathcal{B})$, then~$t$
  has a neighborhood~$V$ such that $S_u \in \mathcal{B}$ for all $u\in V,$ by
  continuity. Hence $X_{n}(t)=1$ for large~$n$. Since we have pointwise
  convergence on a set of (a.s.) full measure, we can apply the
  dominated convergence theorem to conclude
  \[
    \lim_{n\to\infty} \int_0^T X_n(t) dt
    = \int_0^T \mathbf{1}_\mathcal{B}(S_t) dt, \qquad \text{a.s.}
  \]
  But this is the desired result, since
  \begin{align*}
    \int_0^T X_n(t) dt &=
     \sum_{i=1}^n \int_0^T X_{ni}(t) dt \\
     &= \sum_{i=1}^n \int_{\mathcal{T}_i^n} X_{ni}(t) dt \\
     &= \sum_{i=1}^n |\mathcal{T}_i^n|\ \mathbf{1}_{\{S_t\in
       \mathcal{B}\ \forall t\in \mathcal{T}_i^n\}}
     = \frac1n \sum_{i=1}^n  \mathbf{1}_{\{S_t\in \mathcal{B}\ \forall t\in \mathcal{T}_i^n\}}.
  \end{align*}
\end{proof}
Theorem~\ref{thm:corr} suggests the approximation
\begin{equation}\label{eq:appr}
  e^{-rT}\expec(F-A)^+ \approx n e^{-rT} \expec\left(\frac{F}{n} -
    \int_0^{T} \mathbf{1}_\mathcal{B}(S_t) dt\right)^+
\end{equation}
for the price of the structure floor~\eqref{eq:struc floor}.
It is obtained from replacing~$F$ by $F/n$ in the relation
\[
  \expec(nF-A)^+ \sim n \expec\left(F- \int_0^{T} \mathbf{1}_\mathcal{B}(S_t)dt\right)^+,
\]
which follows from Theorem~\ref{thm:corr}.
On the right hand side of~\eqref{eq:appr} we recognize
the price of a put on the occupation time of~$(S_t)$,
also called a corridor option. Fusai~\cite{Fu00} studied such options
in the Black-Scholes model. In particular, his Theorem~1 gives an
expression for the characteristic function of $\int_0^T \mathbf{1}_\mathcal{B}(S_t)dt$.
Since the formula is rather involved, we do not reproduce it here.
Section~4 of Fusai~\cite{Fu00} explains how to compute the corridor
option price from the characteristic function by numerical Laplace inversion.
%Just note that his parameters $l,u,m$ are linked to ours by
%\[
%  l = \frac{1}{\sigma} \log(B_{\mathrm{low}}/S_0), \qquad
%  u = \frac{1}{\sigma} \log(B_{\mathrm{up}}/S_0), \qquad
%  m = \frac{1}{\sigma}\left(r - \frac{\sigma^2}{2}\right).
%\]

This approximation holds for period lengths tending to zero.
One could also let the number of coupons tend
to infinity for a fixed period length~$P$, so that maturity increases linearly
with~$n$. The dependence of the random variables~$C_i$ and~$C_j$
decreases for large $|i-j|$, and so we conjecture a central limit theorem, i.e.,
that
\[
  \frac{A - \expec[A]}{\sqrt{\mathbf{Var}[A]}}
\]
converges in law to a standard normal random variable as $n\to\infty$. Note that
$\expec[A]=\sum_{i=1}^n \expec[C_i]$ and $\expec[A^2]
=\expec[A] + 2 \sum_{i<j} \expec[C_i C_j]$ can be easily
computed from Proposition~\ref{prop:one per} respectively Theorem~\ref{thm:main}.
The structure floor~\eqref{eq:struc floor} could then be
approximately valuated by a Bachelier-type put price formula.
We were not able, though, to verify any of the mixing conditions~\cite{Br05}
that could lead to a central limit result. This is therefore left for
future research.

\bibliographystyle{siam}
\bibliography{barriers}

\end{document}